\documentclass[reprint,
nofootinbib,
amsmath,amssymb,
aps,
pre,
longbibliography
]{revtex4-2}

\usepackage{graphicx}\usepackage{dcolumn}\usepackage{bm}\usepackage{hyperref}

\usepackage{amsthm}
\usepackage{amsfonts}
\usepackage{graphicx}
\usepackage{epstopdf}

\usepackage{amsopn}

\usepackage{siunitx}[=v2]
\usepackage{booktabs}
\usepackage{paralist}

\usepackage{nicefrac}
\usepackage{cleveref}

\usepackage{xcolor}

\usepackage[caption=false]{subfig} \crefname{hypothesis}{Hypothesis}{Hypotheses}

 \ifpdf
  \DeclareGraphicsExtensions{.eps,.pdf,.png,.jpg}
\else
  \DeclareGraphicsExtensions{.eps}
\fi

\newcommand\boxprod{\mathbin{\text{\scalebox{.84}{$\square$}}}}

\newtheorem{theorem}{Theorem}

\crefname{theorem}{Theorem}{Theorems}
\Crefname{theorem}{Theorem}{Theorems}

\crefname{figure}{Fig.}{Figs.}
\Crefname{figure}{Fig.}{Figs.}

\crefname{section}{Section}{Sections}
\Crefname{section}{Section}{Sections}

\crefname{appendix}{Appendix}{Appendices}
\Crefname{appendix}{Appendix}{Appendices}

\crefname{table}{Table}{Tables}
\Crefname{table}{Table}{Tables}

\usepackage[draft, layout={inline},author=]{fixme}

\FXRegisterAuthor{np}{enp}{\color{magenta}[NP]}
\FXRegisterAuthor{df}{edf}{\color{blue}[DF]}
\FXRegisterAuthor{xs}{exs}{\color{red}[XS]}

\makeatletter
\renewcommand*\FXLayoutInline[3]{{\@fxuseface{inline}\ignorespaces#3: #2}}
\makeatother

\ifpdf
\hypersetup{
  pdftitle={An Example Article},
  pdfauthor={N. Pitsianis, D. Floros, and X. Sun}
}
\fi

\begin{document}

\title{The Fiedler connection to the parametrized modularity
  optimization for community detection}

\author{Dimitris Floros}
\affiliation{Nicholas School of the Environment\\
        Duke University, Durham NC 27708, USA}

\author{Nikos Pitsianis}
\altaffiliation[Also at the ]{School of Electrical and Computer Engineering\\
  Aristotle University of Thessaloniki, Thessaloniki, 54124, Greece
}

\author{Xiaobai Sun}
\affiliation{Department of Computer Science\\
        Duke University, Durham NC 27708, USA}

\date{\today}

\begin{abstract}
  This paper presents a comprehensive analysis of the generalized spectral
  structure of the modularity matrix $B$, which is introduced by Newman as the
  kernel matrix for the quadratic-form expression of the modularity function $Q$
  used for community detection.  The analysis is then seamlessly extended to the
  resolution-parametrized modularity matrix $B(\gamma)$, where $\gamma$ denotes
  the resolution parameter. The modularity spectral analysis provides fresh and
  profound insights into the $\gamma$-dynamics within the framework of
  modularity maximization for community detection. It provides the first
  algebraic explanation of the resolution limit at any specific $\gamma$ value.
  Among the significant findings and implications, the analysis reveals that (1)
  the maxima of the quadratic function with $B(\gamma)$ as the kernel matrix
  always reside in the Fiedler space of the normalized graph Laplacian $L$ or
  the null space of $L$, or their combination, and (2) the Fiedler value of the
  graph Laplacian $L$ marks the critical $\gamma$ value in the transition of
  candidate community configuration states between graph division and
  aggregation. Additionally, this paper introduces and identifies the Fiedler
  pseudo-set (FPS) as the de facto critical region for the state transition.
  This work is expected to have an immediate and long-term impact on
  improvements in algorithms for modularity maximization and on
  model transformations.
\end{abstract}
 
\maketitle

\section{Introduction}
\label{sec:introduction}

Denote by $G(V,E)$ a graph $G$ with vertex set $V$ and edge set
$E$. Graph $G$ may represent a complex network in the real world that
exhibits non-homogeneous, non-trivial connectivity patterns and
indicates the co-existence of intrinsic structures, substructures, and
inexorable randomness. Detecting underlying community structures
within a network is important to identifying and understanding the
functions and functionalities across various parts of the network, as
well as their collective behaviors or properties in certain aspects,
such as propagation, synchronization, phase transition, resilience to
perturbation or adversarial
changes~\cite{desollaprice1965,barabasi1999,newman2001,girvan2002,newman2006,rosvall2008,newman2011,barabasi2016,newman2018}.

For community detection, which is synonymous with graph clustering, we
attempt to find an underlying or intrinsic community/cluster
configuration by a certain model.  Among other models, the modularity
model by Girvan and Newman~\cite{girvan2002} stands as one of the most
influential in academic research and in practical implementation
and deployment. Upon the discovery of the resolution limit with the
model~\cite{fortunato2007}, the generalized modularity model was
introduced, which is equipped with the resolution parameter,
conventionally denoted as $\gamma$~\cite{reichardt2006}. The
generalized model gave rise to advanced issues related to parameter
tuning, learning and validating, and subsequent development of various
multi-resolution modularity
models~\cite{sheikholeslami1998,fortunato2007,schaub2012,traag2013,kawamoto2015,chen2018,veldt2019}.
The purpose of this work is to advance our understanding of the inner
working mechanism within the modularity maximization framework.

We use the following basic assumptions and notations throughout the
paper.  Graph $G(V,E)$ is undirected, connected, with $n=|V|$ nodes
linked by $m=|E|$ edges. The edges are either unweighted or weighted
non-negatively. 
Denote by $A$ the adjacency matrix of graph $G$, $A\geq 0$.
Denote by $e$ the constant-1 vector. Denote by $d$ the degree vector, 
$d = Ae$.\footnote{We use $d$ to denote the degree vector, following the convention in the
  literature on graph theories and applications.}
Denote by $\Omega$ a configuration of communities or clusters on $G$,
$\Omega = \{ C_i \}$, where $C_i$ are connected subgraphs that cover
the entire vertex set $V$.  We focus on the case that the clusters in
configuration $\Omega$ do not overlap.

According to the modularity model, we search for a configuration
$\Omega_{\max}$ with the maximal modularity measure over all possible
cluster configurations:
\begin{equation}
  \label{eq:Q-maximization} 
  \Omega_{\max} = \arg \max_{\Omega} Q(\Omega).  
\end{equation}
The modularity function $Q$ is defined as follows, assuming
$G$ is unweighted,
\begin{equation}
  \label{eq:Q-original}
  Q(\Omega) = \frac{1}{2m} \sum_{C\in \Omega }
  \sum_{i,j\in C} \left(  A(i,j) - \frac{d(i)d(j)}{2m} \right), 
\end{equation}
where $d(i)$ is the degree of vertex $v_i$. 
By maximizing the $Q$-measure, one extracts a desirable community
property that the number of intra-community links exceeds what would
be expected by random chance among the graphs with the same degree
distribution.  For a graph with non-negatively weighted
edges, we simply replace $2m$ by $e^{\rm T}Ae$, the total volume (sum)
of $A$.
Without loss of generality, we scale $A$ so that the total volume of
$A$ is $1$, $ e^{\rm T}A e = 1$. By this scaling, the degree vector
$d$ becomes a stochastic vector, $e^{\rm T} d=1$.

An alternative expression of the modularity function can be arrived
from different modeling principles~\cite{reichardt2006},
\begin{equation}
  \label{eq:Q-intra-inter-terms}
  Q(\Omega) = \sum_{C\in \Omega }
  \alpha( C,C) - \alpha^2(C, V), 
\end{equation}
where, 
\begin{equation}
  \label{eq:Q-alpha}
  \alpha(C,S) \triangleq  e^{\rm T}A(C,S)e
  = \sum_{\substack{i\in C\\j\in S}} A(i,j),
  \quad  C , S \subseteq V.    
\end{equation}
In particular, $\alpha(C,V)$ is the volume of cluster $C$ on graph
$G$, $\alpha( \{v_i\},V) = d(i)$ is the (weighted) degree of vertex
$v_i$, and $\alpha(C,V) = \sum_{i\in C} d(i)$. We have scaled $A$ so 
that $\alpha(V,V)=1$.  In the expression
(\ref{eq:Q-intra-inter-terms}) at the cluster level, the first term 
contains only the intra-cluster connections, and the second term 
includes the interconnections as well.

In order to characterize the behavior of the modularity model and to
develop algorithmic rules for cluster merges and splits in the search
for an optimum configuration $\Omega_{\max}$ of
(\ref{eq:Q-maximization}), one first investigates the fundamental case
that $\Omega$ has at most two clusters.
Any two-cluster configuration $\Omega = \{C_1, C_2\}$ can be
represented by an indicator $s = s(\Omega)$: $s(i) = 1$ for $i\in C_1$
and $s(j) = -1$ for $j\in C_2$. Newman introduced the {\em modularity
  matrix} $B$ and expressed the modularity $Q$ as the quadratic
function in $s$ with kernel $B$,
\begin{equation}
  \label{eq:Q-formulation-B-kernel}
  B = A - d \cdot d^{\rm T},
  \quad 
  Q(s) = \frac{1}{2} s^{\rm T}\,  B\,  s,
  \quad
  s \in \{1,-1 \}^{n}.
\end{equation}
Denote by $\Omega_{\vee}$ the all-in-one configuration, i.e.,
$s(\Omega_{\vee}) = e$. Clearly, $e^{\rm T}B \, e = 0$.  If there
exists any configuration $\Omega$ with
$s(\Omega)^{\rm T}B s(\Omega) > 0$, then the division by $\Omega$ is
favored over the all-in-one configuration by the modularity
maximization (\ref{eq:Q-maximization}).

With the quadratic-form expression of $Q$ as in
(\ref{eq:Q-formulation-B-kernel}), Newman established a bridge to algebraic
graph partitions by the graph Laplacian. Newman noted in~\cite{newman2006} the
resemblance in the formats. In this paper, we provide in
\cref{sec:fiedler-connection} a precise description of the connection between
the modularity matrix $B$ and the normalized graph Laplacian $L$.
In \cref{sec:resolution-limit} we give an algebraic explanation for the
resolution limit with the original modularity model.
In \cref{sec:fiedler-transition}, we express $Q(\Omega \mid \gamma)$, at any value of
$\gamma$, as the quadratic function with $B(\gamma)$ as the kernel matrix,
$B(\gamma)$ is the parametrized modularity matrix.  We relate the spectral
structure of $B(\gamma)$, across all $\gamma$-variation, to the spectral
structure of the same graph Laplacian $L$. In particular, $B(\gamma)$ and $L$
share the same invariant subspace associated with the Fiedler value, which is
the smallest nonzero eigenvalue of the normalized graph Laplacian $L$, often
denoted as $\lambda_2$.
In \cref{sec:FPS-transition}, we present the Fiedler pseudo-set (FPS), denoted
as $\Lambda_2(\varepsilon)$, as the necessary extension of the Fiedler value
$\lambda_2$ to complex networks, where $\varepsilon$ is a perturbation
tolerance.
In \cref{sec:subgraph-division}, we describe the sub-modularity matrix for any
subgraph, preserving its topological and probabilistic interconnection to the
graph complement. This effectively extends our modularity spectral analysis to
community configurations with multiple communities.

Our key findings include the following.
\begin{inparaenum}[(i)]
\item The modularity matrix $B(\gamma)$ functions as a graph spectral modulation
and selection mechanism.  At any fixed value of $\gamma$, the mechanism puts all
graphs into two major categories, depending on whether their Fiedler values are
above or below $\gamma$.
\item There exists a state transition in the effects of the spectral modulation
and selection on each graph. The Fiedler value $\lambda_2$ marks the critical
point: when $\gamma < \lambda_2$, $B(\gamma)$ advocates keeping the graph
undivided; when $\gamma > \lambda_2$, $B(\gamma)$ promotes a graph split.
We have thus pinpointed a systematic source that is responsible for the
resolution limit that exists at any fixed value of $\gamma$, not only at
$\gamma=1$ for the original modularity model.
\item Instead of a single critical point, the Fiedler pseudo-set (FPS)
$\Lambda_2(\varepsilon)$ is the de facto critical region for the state
transition. The FPS serves several other graph/network analysis tasks as well,
including connectivity sensitivity analysis.
\end{inparaenum} 

We present several case studies with numerical results throughout the paper. We
give additional comments in \cref{sec:additional-implications}.

\section{The Fiedler connection}
\label{sec:fiedler-connection}

We present a rigorous analysis of the spectral structure of the
modularity matrix $B$ in terms of the generalized eigenvalue
decomposition: $Bx_j = \beta_j Dx_j$, $j=1,\dots,n$, where
$D = \mbox{\rm diag}(d)$ is the diagonal matrix with the degree vector
$d$ on the diagonal.  By the assumption that $G$ is connected, $d>0$.
Equivalently, we consider the standard eigenvalue decomposition,
$\hat{B} u_j = \beta_j u_j$, where
\begin{equation}
  \label{eq:Bhat-definition} 
  \hat{B} = D^{-1/2} B D^{-1/2}. 
\end{equation}
We may refer to $\hat{B}$ as the standard modularity matrix; to $B$,
as the prime one. When there is no ambiguity, we may simply call each the
modularity matrix.
Matrix $\hat{B}$ connects the normalized graph Laplacian
{$L = I - D^{-1/2} A D^{-1/2}$} and the modularity function $Q(s)$ as
follows.
We have on the one side,
\begin{equation}
  \label{eq:Bhat-Lhat-connection}
   \hat{B} = (I - L)  - d^{1/2}(d^{1/2})^{\rm T}, 
 \end{equation}
and on the other side,
\begin{equation}
  \label{eq:Q-in-Bhat-kernel}
  \begin{aligned} 
    Q(s\!\mid\! B)
    & = \, Q(y(s)  \!\mid\! \hat{B} )
      = y^{\rm T}(s)  \hat{B}\, y(s),
  \\
  & \phantom{xxx}  y(s) = D^{1/2} s,
  \quad s \in \{1,-1\}^n. 
 \end{aligned}
\end{equation}
The variable change from $s$ to $y(s)$ in (\ref{eq:Q-in-Bhat-kernel}) is
consistent with the conversion of the generalized eigenvalue problem with
$(B,D)$ to the standard eigenvalue problem with $\hat{B}$.

Clearly, for any $s\in \{1,-1\}^n$, $ {\rm sign}(y(s))=s $, and $y(s)$
is of unit length, i.e.,
$y(s)^{\rm T}y(s) = s^{\rm T}D s = e^{\rm T}d =1$.
That is, the discrete set
$ Y= \{y(s)=D^{1/2}s \!\mid\! s \in \{1,-1\}^{n}\}$ is a subset of
the unit sphere $S = \{u \!\mid\! u^{\rm T}u = 1, u\in \mathbb{R}^{n}\}$.
Thus, the quadratic function on $Y$ can be readily extended to $S$.
In order to gain insight into the discrete maximization of $Q$ over the
set $Y$, we intend to first study the continuous counterpart,
\begin{equation}
  \label{eq:Bhat-Rayleigh-quotient} 
  u_{\max} = \arg \max_{ u\in S}  R(\, u \!\mid\! \hat{B}),
  \quad 
  R(\, u \!\mid\! \hat{B}) \triangleq u^{\rm T} \hat{B} u.  
\end{equation}
The quadratic function $R(\, u \!\mid\!  \hat{B})$ on the unit sphere
$S$ is the Rayleigh quotient of matrix $\hat{B}$.  By the
Courant-Fischer min-max principle~\cite{horn2012}, the maximal value
of $R(\, u \!\mid\!M)$ with any symmetric matrix $M$ is the largest
eigenvalue of $M$, any maximum of $R(\,u \!\mid\! M)$ is an
eigenvector of $M$ associated with the largest eigenvalue.
Denote by $\beta_{\max}$ the maximal eigenvalue of $\hat{B}$.  If
$\beta_{\max}$ is simple (of multiplicity $1$), the maximum of the
Rayleigh quotient $R(\, u \!\mid\!  \hat{B})$ is unique, up to a sign
flip.  If the multiplicity of $\beta_{\max}$ is greater than $1$, then
any normalized vector in the invariant subspace associated with
$\beta_{\max}$ is a maximum of $R(\, u \!\mid\!  \hat{B})$.  Any
maximum $u_{\max}$ on the unit sphere $S$ can be mapped to
$Q( y(s) \!\mid\! \hat{B})$ or $Q( s \!\mid \!  B)$, with
$s = \mbox{sign}(u_{\max})$.
\begin{theorem}
  \label{thm:Lhat-Bhat-extreme-eigenpairs}
  Let $G(V,E)$ be a connected, undirected graph with adjacency matrix
  $A$. Denote by $d$ the degree vector scaled to be stochastic.
  Denote by $L$ the normalized Laplacian of $G$.
Let $\{\lambda_i\} $ be the Laplacian eigenvalues in non-descending
  order: $\lambda_1 \leq \lambda_2 \leq \cdots \leq \lambda_n$.
Let $\hat{B}$ be the standard modularity matrix as defined in
  (\ref{eq:Bhat-definition}).  Denote by ${\beta}_{\max}$ the largest
  eigenvalue of $\hat{B}$.
Then, the following hold true. 
\setlength{\leftmargini}{2pt}
  \begin{enumerate}
  \item The smallest Laplacian eigenvalue $\lambda_1$ is equal to $0$,
    and it is simple, i.e., $\lambda_1< \lambda_2$. The unique and
    normalized eigenvector associated with $\lambda_1$ is $d^{1/2}$.
\item Any Laplacian eigenvector is an eigenvector of $\hat{B}$, and
    vice versa.
\item The eigenvalues of $\hat{B}$ are
    $\{0\} \cup \{ 1-\lambda_i \mid i=2,\dots, n \}$.  In particular,
    \begin{equation}
      \label{eq:beta-max} 
         {\beta}_{\max} = \max\{0, 1-\lambda_2\}.
      \end{equation} 
      Any eigenvector $u_{\max}$ associated with $\beta_{\max}$ is
      characterized as follows,
\setlength{\leftmargini}{2pt}
    \begin{enumerate} 
    \item Case $\lambda_2 > 1$: $\beta_{\max} = 0$,
      $u_{\max} = d^{1/2} > 0 $, and $u_{\max}$ is unique.
\item Case $\lambda_2 < 1$: $\beta_{\max} = 1 - \lambda_2 > 0
      $, $u_{\max}$ is orthogonal to $d^{1/2}$, $e^{\rm
        T}u_{\max}=0$. Specifically,
      $u_{\max} \in {\rm span}\{u_i \mid Lu = \lambda_2
      u\}$. When and only when $\lambda_2$ is simple,
      $u_{\max}$ is unique.
\item Case $\lambda_2 = 1$: $\beta_{\max} = 0$,
      $u_{\max} \in {\rm span}\{d^{1/2}, u\mid Lu = u
      \}$.  The null space of $\hat{B}$ is at least of dimension
      $2$.  The elements of any maximum vector
      $u_{\max}$ either have the same sign or have mixed signs.
  \end{enumerate}
\end{enumerate}
\end{theorem}
\begin{proof} 
  By the graph Laplacian properties, $\lambda_1=0$ for any graph, and
  $\lambda_2>0$ for any connected graph, i.e., $\lambda_1$ is simple.
  One can verify directly that $L d^{1/2} =0$ and
  $(d^{1/2})^{\rm T}d^{1/2} = e^{\rm T}d = 1$. Thus, the outer-product
  term $d^{1/2}(d^{1/2})^{\rm T}$ is the orthogonal projector onto the
  invariant subspace associated with $\lambda_1$.  Based on this fact,
  it is rather straightforward to verify the claims in Part 2 and Part 3.
\end{proof}

The smallest nonzero eigenvalue $\lambda_2$ of the normalized Laplacian $L$
plays a critical role in the analysis of the maxima of the quadratic function
$R(\, y \!\mid\! \hat{B})$.  It is known as the {\em Fiedler
value}~\cite{fiedler1973,chung1997} in spectral graph theory and literature. It
is also known as the algebraic connectivity index, as it is closely related to
the Cheeger cut ratio in combinatorial graph theory.
When $\lambda_2$ is small, it indicates that there exists a small set of cut edges that, upon removal, decouple the graph into two subgraphs $C_1$ and $C_2$. Such a graph cut is then found in the sign pattern of an associated eigenvector $u_2$, $i\in C_1$ if $u_2(i) > 0$, and $i\in C_2$ otherwise. Only when $\lambda_2$ is simple, the eigenvector $u_2$ is unique, so is the cut.
When the multiplicity of $\lambda_2$ is greater than one, we define the invariant subspace associated with $\lambda_2$ as the {\em Fiedler subspace}. Any vector $u$ in the Fiedler subspace is a maximum of the Rayleigh quotient $y^{\rm T}\hat{B}y$ over the unit sphere. Orthogonal to $d^{1/2}$, $u$ has mixed signs and induces a graph cut, $s = \mbox{\rm sign}(u)$, which we may refer to as a {\em Fiedler cut}.

The $n$-nodes clique graph $K_n$, $n>1$, is a special example of the case (3-a) in \cref{thm:Lhat-Bhat-extreme-eigenpairs}. The Fiedler value $\lambda_2$ of $K_n$ is equal to $n/(n-1)$, greater than $1$.  By \cref{thm:Lhat-Bhat-extreme-eigenpairs}, the maximum eigenvector $u_{\max}$ of $\hat{B}$ is $d^{1/2} = \nicefrac{e}{\sqrt{n}}$, positive on all vertices, the clique is not divided.
Real-world networks tend to be large and sparse. Many of them fall into the case (3-b), $\lambda_2 < 1$, as the average of the nonzero Laplacian eigenvalue is $\mbox{\rm trace}(L)/(n-1) = n/(n-1)$.  By \cref{thm:Lhat-Bhat-extreme-eigenpairs}, such graphs are subject to cuts indiscriminately. We address this concerning issue next.

\section{The resolution limit}
\label{sec:resolution-limit}

\Cref{thm:Lhat-Bhat-extreme-eigenpairs} identifies, from the
spectral perspective, a systematic source responsible for what is
known as the resolution limit with the modularity model
(\ref{eq:Q-original})~\cite{fortunato2007}.
The resolution limit phenomena include systematic errors in cluster
merges or splits. Denote by $G_{\textrm{RoK}(p,q)}$ the graph composed
of $q$ copies of $K_p$ linked by $q$ edges circularly, as illustrated
in \cref{fig:fiedler-cut-plots:rok}. We may refer to such a graph as a
ring of cliques (RoK). {Fortunato and Berthelemy} first noted
in~\cite{fortunato2007} that in the optimum configuration
$\Omega_{\max}$ by the modularity model, the neighbor cliques of
$G_{\textrm{RoK}(p,q)}$ are incorrectly grouped together.
We show in this section that the modularity maximization
(\ref{eq:Q-original}) on a connected graph $G$ with $\lambda_2 <1$ may
systematically lead to a false-positive detection of more than one
clusters/communities. We demonstrate this type of community detection
errors on a few graphs that have homogeneous or nearly homogeneous
connectivity structures. We trace the fault to the fixed spectral
modulation with the eigen-projector $d^{1/2}(d^{1/2})^{\rm T}$.

\begin{figure*}[thbp]
\hspace*{\fill}
  \hspace*{-1em}
\subfloat[][Cube(4)\label{fig:fiedler-cut-plots:cube}]{ \centering
    \includegraphics[height=0.12\textheight]{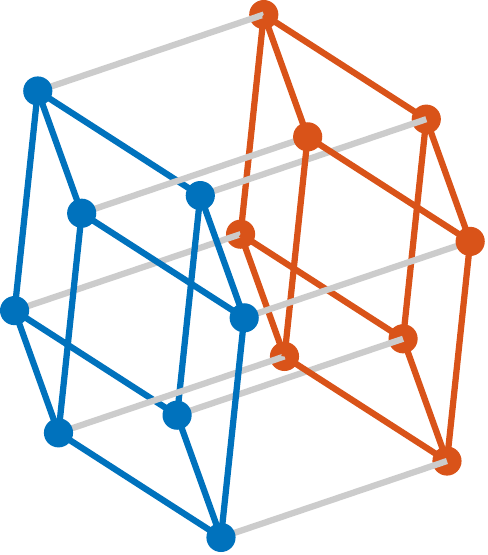}
  }
\hspace*{\fill}
\subfloat[][Buckyball\label{fig:fiedler-cut-plots:buckyball}]{
    \centering
    \includegraphics[height=0.11\textheight]{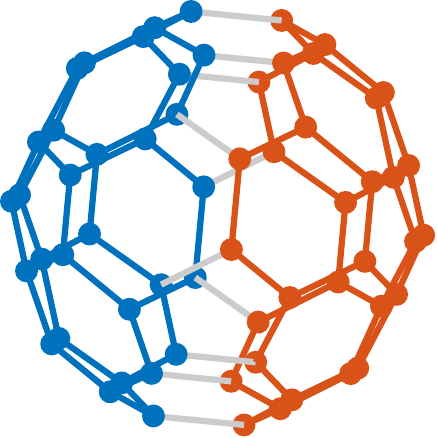}
  } 
\hspace*{\fill}
\subfloat[][$P \boxprod C(10,10)$\label{fig:fiedler-cut-plots:mesh}]{
    \centering
    \includegraphics[height=0.10\textheight]{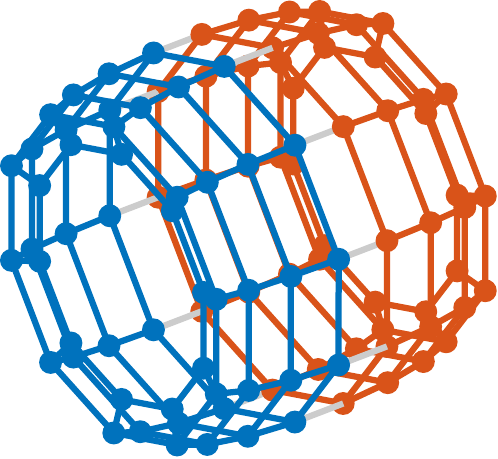}
    }
\hspace*{\fill}
\subfloat[][RoK(6,5)\label{fig:fiedler-cut-plots:rok}]{
    \centering 
    \includegraphics[height=0.10\textheight]{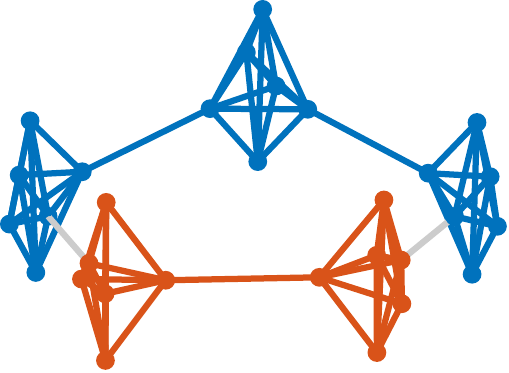}
  }
\hspace*{\fill}
\caption{The Fiedler cuts on four graphs: the 4-dimensional hypercube
    $G_{\textrm{cube}(4)} $ in (a), the buckyball $G_{\rm buckyball}$ in (b),
    the cylindrical mesh $G_{P\boxprod C(10,10)}$ in (c) and the ring of cliques
    $G_{\textrm{RoK}(6,5)}$ in (d).  The nodes on each graph are homogeneous or
    nearly homogeneous.  On each graph, the subgraphs by a Fiedler cut are shown
    in two different colors. The modularity maximization at $\gamma=1$ favors
    the Fiedler cut over no-cut on each of the graphs. Detailed descriptions and
    information on the graphs are in Section~\ref{sec:resolution-limit} and in
    Table~\ref{tab:fps} }
\label{fig:fiedler-cut-plots}
\end{figure*}

{\tt Hypercube $G_{\textrm{cube}(k)}$}, of dimension $k>2$.
The graph has $2^k$ nodes with regular degree $k$, as illustrated in
\cref{fig:fiedler-cut-plots:cube}.
Unequivocally, the hypercube is not to be decoupled. The Fiedler value
$\lambda_2$ of $G_{\textrm{cube}(k)}$ is equal to {$2/k$} with
multiplicity $k$.  By Theorem~\ref{thm:Lhat-Bhat-extreme-eigenpairs},
$\beta_{\max} = 1 - \lambda_2 >0$ and there are many different Fiedler
cuts. Any $u_{\max}$ is in the Fiedler space of dimension $k$ and
induces a cut to the hypercube.  Let $s = {\rm sign}(u_{\max})$. One
can verify that $Q(s \!\mid\!  B ) = s^{\rm T}B s>0$. In particular,
one of the Fiedler cuts leads to two subcubes of dimension $k\!-\!1$,
see a Fiedler cut of $G_{\textrm{cube}(4)}$ illustrated in
\cref{fig:fiedler-cut-plots:cube}.  The modularity value $Q$ at this
cut is $(k\!-\! 2)/(2k)$, greater than the value at the all-in-one
configuration $\Omega_{\vee}$, $Q(\Omega_{\vee})=0$. The modularity
maximization thereby favors any of the Fiedler cuts over no-cut,
falsely.

{\tt Buckyball $G_{\textrm{buckyball}}$\, (Molecule $C_{60}$)}.  The
 graph has $60$ nodes with regular degree $3$, as shown in
 \cref{fig:fiedler-cut-plots:buckyball}. The Fiedler value $\lambda_2$
 of the buckyball is of multiplicity $3$, analytically. Numerically,
 $\lambda_2 \approx 0.081$, much smaller than $1$.  Thus, any maximum
 eigenvector $u_{\max}$ of $\hat{B}$ is in the Fiedler space of
 dimension $3$, and induces a cut to the Buckyball graph. The
 modularity value at such a cut is 0.39, greater than zero at
 $\Omega_{\vee}$.

{\tt Cylindrical mesh $G_{P\boxprod C(p,q)}$.}  The mesh is the
Cartesian product of the $p$-nodes path graph and the $q$-nodes cycle
graph, i.e., $G_{P\boxprod C(p,q)} = P_{p} \boxprod C_{q}$.  Although
its nodes have nearly homogeneous connection structures, the mesh is
subject to a cut by the modularity maximization because $\lambda_2 <
1$.  The cut is unique as $\lambda_2$ is simple. The cylindrical mesh
$G_{P\boxprod C(10,10)}$ with the cut is illustrated in
\cref{fig:fiedler-cut-plots:mesh}.  For $G_{P\boxprod C(10,10)}$,
$\lambda_2 \approx 0.027$, much smaller than $1$.  The modularity
$Q$-value at the Fiedler cut is 0.447, greater than zero at
$\Omega_{\vee}$.

\section{The critical transition value}
\label{sec:fiedler-transition}

Among other approaches to overcoming the resolution limit, Reichardt
and Bornholdt introduced the resolution parameter
$\gamma$~\cite{reichardt2006}, using the expression
(\ref{eq:Q-intra-inter-terms}),
\begin{equation}
  \label{eq:Q-ngrb-gamma}
  Q(\Omega \mid \gamma) = \sum_{C\in \Omega}
  \alpha(C,C) - \gamma \, \alpha^{2}(C,V),
  \quad \quad \gamma > 0. 
\end{equation}
This $\gamma$-parametrized model is known as the generalized
modularity model. Roughly speaking, a smaller $\gamma$ value promotes
coarser clusters, whereas a larger $\gamma$ value advocates finer
clusters.

For cluster configurations with at most two clusters, the parametrized
model can be written in quadratic forms similar to those in
(\ref{eq:Bhat-definition}) and (\ref{eq:Bhat-Lhat-connection}). We
parametrize the prime and standard modularity matrices $B$ and
$\hat{B}$, respectively, as follows,
\begin{equation} 
  \label{eq:B-Bhat-gamma}
\begin{aligned} 
B(\gamma) & =   A - \gamma \, d\, d^{\rm T},
\\ 
\hat{B}(\gamma) & =  D^{-1/2}B(\gamma) D^{-1/2}
\\
&= (I - L) - \gamma\, d^{1/2}\, (d^{1/2})^{\rm T},
\end{aligned}
\end{equation}
where $L$ is the normalized graph Laplacian.
Using the same notation and derivation as in the preceding sections,
we have
\begin{equation}
  \label{eq:Q-B-gamma}
  \begin{aligned}   
  Q(s \!\mid\! B(\gamma))
  & = \frac{1}{2}
     \left( s^{\rm T} B(\gamma) \, s + (1-\gamma) \right), 
  \\
  & \phantom{xxx} s \in \{1,-1\}^{n}.
  \end{aligned} 
\end{equation}
and 
\begin{equation}
  \label{eq:Q-Bhat-gamma}
  \begin{aligned} 
  2\, Q( \, y(s) \!\mid\! \hat{B}(\gamma))
  &
  = y(s)^{\rm T} \hat{B}(\gamma) \, y(s) + (1-\gamma),
  \\ 
  y(s) &  = D^{1/2}s,
  \quad s \in \{1,-1\}^{n}.
\end{aligned}
\end{equation}
We provide in \cref{sec:appendix-proofs} a detailed proof of
(\ref{eq:Q-B-gamma}).
As $y(s)$ is fittingly on the unit sphere $S$, we extend the discrete
support of $Q$ to $S$,
\begin{equation}
  \label{eq:Rayleigh-Bhat-gamma}
  \begin{aligned}
    u_{\max} & =
    \arg \max_{u \in S} R( \, u \!\mid\! \hat{B}(\gamma)),
    \\
    & \phantom{xx}
    R( \, u \!\mid\! \hat{B}(\gamma))
  = 
  u^{\rm T} \hat{B}(\gamma) \, u.
\end{aligned}
\end{equation}

We generalize Theorem~\ref{thm:Lhat-Bhat-extreme-eigenpairs} from the
particular case of
$\gamma=1$ to arbitrary values of the resolution parameter $\gamma$.
\begin{theorem}
  \label{thm:Lhat-Bhat-across-gamma}
  Let $\hat{B}(\gamma)$ be the $\gamma$-parametrized standard
  modularity matrix as defined in (\ref{eq:B-Bhat-gamma}),
  $\gamma > 0$. Let $\lambda_{i}$ and $\beta_{\max}$ be specified as
  in Theorem~\ref{thm:Lhat-Bhat-extreme-eigenpairs}.
Then, the following hold true.
\setlength{\leftmargini}{2pt} 
  \begin{enumerate}
  \item The eigenvalues of $\hat{B}(\gamma)$ are
    $\{ 1 -\gamma \} \cup \{1\!-\!\lambda_i, i=2,\dots, n\}$.  The
    corresponding eigenvectors of $\hat{B}(\gamma)$ are the Laplacian
    eigenvectors associated with $\lambda_i$, $i=1,2,\!\dots,\! n$,
    respectively, and therefore $\gamma$-independent.
\item The largest eigenvalue of $\hat{B}(\gamma)$ depends on $\gamma$,
  \begin{equation}
    \label{eq:Bhat-beta-max}
    \beta_{\max} = \max\{ 1-\gamma, 1-\lambda_2\}. 
  \end{equation}
Any eigenvector $u_{\max}$ associated with $\beta_{\max}$ is
  characterized as follows,
\begin{enumerate}\item Case $\gamma < \lambda_2$: $\beta_{\max} = 1 - \gamma$,
    $u_{\max} = d^{1/2} > 0 $, and $u_{\max}$ is unique.
  \item Case $\gamma > \lambda_2$: $\beta_{\max} = 1 -
    \lambda_2$, $u_{\max}$ is orthogonal to $d^{1/2}$,
    $e^{\rm T}d^{1/2} = 0$, and specifically,
    $u_{\max} \in {\rm span}\{u \mid Lu = \lambda_2
    u\}$.  The maximum vector is unique if and only if
    $\lambda_2$ is simple.
  \item Case $\gamma = \lambda_2$:
    $\beta_{\max} = 1-\lambda_2=1-\gamma$, and
    $u_{\max} \in {\rm span}\{d^{1/2}, u\mid Lu = \lambda_2 u
    \}$, i.e., $u_{\max}$ is in the invariant subspace of $L$
    associated with $\lambda_1$ and $\lambda_2$. The elements of
    $u_{\max}$ either have the same sign or have mixed signs.
  \end{enumerate}
\end{enumerate}
\end{theorem}

We revisit the graphs $G_{\rm cube(k)}$, $G_{\rm buckyball}$ and
$G_{P\boxprod C(p,q)}$, which are described in
\cref{sec:resolution-limit} and illustrated in \cref{fig:fiedler-cut-plots}.
For each of the graphs, the Fiedler value $\lambda_2$ is less than
$1$, and by \cref{thm:Lhat-Bhat-across-gamma}, the $\gamma$-value
shall be set less than $\lambda_2$.  At the all-in-one configuration
$\Omega_{\vee}$, we have $Q(\Omega_{\vee} \mid \gamma) = 1 - \gamma$,
which is equal to $\beta_{\max}$ when and only when
$\gamma \leq \lambda_2$.
For each graph, by setting $\gamma$ smaller than $\lambda_2$, the
graph is kept in $\Omega_{\max}$ as a whole as expected.
We conclude that the fault with the original modularity model lies in
the fixed setting of $\gamma$ to $1$, oblivious and non-adaptive to
graph structures.

Theorem~\ref{thm:Lhat-Bhat-across-gamma} leads to the following
revelation.
\begin{theorem}[The critical transition at the Fiedler value]
  \label{thm:Fiedler-critical-transition}
For any undirected, connected graph $G$, the Fiedler value
  $\lambda_2$ is the critical value of the resolution parameter
  $\gamma$ in the following sense: when $\gamma < \lambda_2$, the
  maximum eigenvector of $\hat{B}(\gamma)$ is unique, its elements are
  of the same sign, not inducing any graph division; when
  $\gamma > \lambda_2$, the elements of any maximum eigenvector of
  $\hat{B}(\gamma)$ have mixed signs, inducing a graph split.
\end{theorem}

A few remarks are in order. For graph partitions with the minimum set
of cut edges, the Fiedler value is a measure of the min-cut edge set
relative to the volumes of the divided subgraphs. For graph
clustering/community detection, by
\cref{thm:Fiedler-critical-transition}, the Fiedler value assumes an
additional important role as the critical value of the resolution
parameter $\gamma$ on its effect on the state of the optimum
configuration $\Omega_{\max}(\gamma)$: with or without a graph
division.
The modularity spectral theory summarized in
\cref{thm:Lhat-Bhat-across-gamma,thm:Fiedler-critical-transition}
advances our understanding of the resolution parameter $\gamma$ in the
model (\ref{eq:Q-ngrb-gamma}). At the forefront, $\gamma$ regulates
the policies regarding graph splits or merges. In the backend, the
resolution parameter effectively modulates the maximum spectral value
of $\hat{B}(\gamma)$.

The first two preponderant implications of the modularity spectral
theory are immediate.
\begin{inparaenum}[(i)]
\item The $\gamma$-parametrized model at any fixed value of $\gamma$,
  $\gamma \neq 1$, inherits the very same flaw as in the case of
  $\gamma= 1$: the maximum vectors of $\hat{B}(\gamma)$
  indiscriminately induce Fiedler cuts on any graph with the Fiedler
  value $\lambda_2$ less than $\gamma$.
\item If $\gamma$ is set greater than $2$, then a Fiedler cut is
  induced on any graph. Consider, for example, any clique $K_n$,
  including the edge graph $K_2$. The modularity $Q$ at $\gamma = 2$ 
  has a higher value at a Fiedler cut than the value at $\Omega_{\vee}$ 
  without graph split. 
\end{inparaenum}
We will discuss in \cref{sec:additional-implications}
additional implications, some of which are subtle.

\section{The critical transition region}
\label{sec:FPS-transition}

\begin{table*}[thbp]
  \centering
\caption{The Fiedler value $\lambda_2$ and the Fiedler pseudo-set (FPS) descriptor
    $(t_{\min}, t_{\max})$ for each of the $4$ small and structured graphs,
    Zachary's karate club~\cite{zachary1977}, and $2$ larger and random graphs
    on the table. The topological structures of the first four graphs are
    illustrated in \cref{fig:fiedler-cut-plots}. Two types of
    perturbations---edge removal ($r_e$) and edge rewiring ($r_w$)---at two
    levels are considered. At the first level, one edge is removed from or
    rewired on each small graph; $1\%$ of the edges, on each large graph.  At
    the second perturbation level, two edges are removed from or rewired on each
    small graph; $2\%$ of the edges, on each large graph.  For each case, the
    FPS is estimated by \num{100} trials. See more detailed descriptions in
    \cref{sec:FPS-transition}. }
\label{tab:fps}
\renewcommand{\arraystretch}{1.2}
\begin{tabular}{
      l
      @{\hskip 2.5em}
      l
      @{\hskip 2.5em}
      S[table-format=1.3, table-number-alignment = right]
      S[table-format=1.3, table-number-alignment = right]
      @{\hskip 2.5em}
c
      c
c
      c
c
      c
c
      c
      }
      \toprule
&  &  &
        & \multicolumn{2}{c}{ $r_{w}=1$ } 
        & \multicolumn{2}{c}{ $r_{e}=1$ } 
        & \multicolumn{2}{c}{ $r_{w}=2$ } 
        & \multicolumn{2}{c}{ $r_{e}=2$ } 
      \\
      \cmidrule(lr){5-6} \cmidrule(lr){7-8} \cmidrule(lr){9-10} \cmidrule(lr){11-12}
      \multicolumn{2}{l}{Graph} 
        & {$\lambda_{2}$} & {$\lambda_{3}$} 
        & {$t_{\min}$} & {$t_{\max}$}
        & {$t_{\min}$} & {$t_{\max}$}
        & {$t_{\min}$} & {$t_{\max}$}
        & {$t_{\min}$} & {$t_{\max}$}
      \\
\midrule
      Cylindrical mesh & $P \boxprod C(33,3)$ & 0.002 & 0.009 & $\nicefrac{1}{1.01}$ & $\nicefrac{3.99}{1}$ & $\nicefrac{1}{1.05}$ & $\nicefrac{1.01}{1}$ & $\nicefrac{1}{1.03}$ & $\nicefrac{5.72}{1}$ & $\nicefrac{1}{1.09}$ & $\nicefrac{1.02}{1}$ 
      \\
      Ring of cliques & $\textrm{RoK}(6,17)$ & 0.003 & 0.003 & $\nicefrac{1}{3.85}$ & $\nicefrac{1}{1.00}$ & $\nicefrac{1}{3.95}$ & $\nicefrac{1.00}{1}$ & $\nicefrac{1}{3.91}$ & $\nicefrac{2.41}{1}$ & $\nicefrac{1}{\infty}$ & $\nicefrac{1.01}{1}$ 
      \\
      Buckyball & $(n = 60\phantom{6}, m = 90)$ & 0.081 & 0.081 & $\nicefrac{1}{1.16}$ & $\nicefrac{1}{1.01}$ & $\nicefrac{1}{1.14}$ & $\nicefrac{1}{1.10}$ & $\nicefrac{1}{1.41}$ & $\nicefrac{1}{1.01}$ & $\nicefrac{1}{1.37}$ & $\nicefrac{1}{1.09}$ 
      \\
      Hypercube & $\textrm{Cube}(8)$ & 0.250 & 0.250 & $\nicefrac{1}{1.01}$ & $\nicefrac{1}{1.01}$ & $\nicefrac{1}{1.01}$ & $\nicefrac{1}{1.01}$ & $\nicefrac{1}{1.02}$ & $\nicefrac{1}{1.01}$ & $\nicefrac{1}{1.02}$ & $\nicefrac{1}{1.01}$ 
      \\
      \midrule[0.1pt]
      Zachary's karate club & $(n = 34\phantom{6}, m = 78)$ & 0.132 & 0.287 & $\nicefrac{1}{1.11}$ & $\nicefrac{1.30}{1}$ & $\nicefrac{1}{1.18}$ & $\nicefrac{1.06}{1}$ & $\nicefrac{1}{1.20}$ & $\nicefrac{1.46}{1}$ & $\nicefrac{1}{1.18}$ & $\nicefrac{1.07}{1}$ 
      \\
      \midrule[0.2pt]
        &  &  &
        & \multicolumn{2}{c}{ $r_{w}=1\%$ } 
        & \multicolumn{2}{c}{ $r_{e}=1\%$ } 
        & \multicolumn{2}{c}{ $r_{w}=2\%$ } 
        & \multicolumn{2}{c}{ $r_{e}=2\%$ } 
      \\
      \cmidrule(lr){5-6} \cmidrule(lr){7-8} \cmidrule(lr){9-10} \cmidrule(lr){11-12}
      Erd\"{o}s-R\'{e}nyi & $\textrm{ER}(n = 2000, \,\, d_{\textrm{avg}} = 16)$ & 0.519 & 0.523 & $\nicefrac{1}{1.00}$ & $\nicefrac{1.00}{1}$ & $\nicefrac{1}{1.01}$ & $\nicefrac{1}{1.00}$ & $\nicefrac{1}{1.01}$ & $\nicefrac{1.01}{1}$ & $\nicefrac{1}{1.01}$ & $\nicefrac{1}{1.01}$ 
      \\
      Barab\'{a}si-Albert & $\textrm{BA}(n = 2000, \,\, d_{\textrm{avg}} = 16)$ & 0.529 & 0.531 & $\nicefrac{1}{1.01}$ & $\nicefrac{1.00}{1}$ & $\nicefrac{1}{1.01}$ & $\nicefrac{1}{1.00}$ & $\nicefrac{1}{1.01}$ & $\nicefrac{1.00}{1}$ & $\nicefrac{1}{1.02}$ & $\nicefrac{1}{1.00}$ 
      \\
      \bottomrule
    \end{tabular}
\end{table*}

 If the Fiedler value $\lambda_2$ is relatively smaller than
 $\mbox{\rm trace}(L)/(n-1)$, which is the average of the non-zero
 Laplacian, it signifies a meaningful graph partition, or indicates a
 possible cluster configuration with more than one clusters.
There is, however, a crucial distinction between graph partition and
 graph clustering: a small value of $\lambda_2$ is insufficient to
 ascertain the presence of more than one clusters/communities.
For a case in point, consider the RoK graph $G_{\textrm{RoK}(6,17)}$
 and the cylindrical mesh $G_{P\boxprod C(33,3)}$, both are shown in
 \cref{fig:fiedler-cut-plots}. The two graphs have about the same
 number of nodes, and their Fiedler values are small and comparable,
 see \cref{tab:fps}.  Despite having a smaller Fiedler value, the
 cylindrical mesh shall not be split, whereas each clique in the RoK
 graph is seen as a closely tied community. The Fiedler values alone
 are insufficient to spell out the difference in connectivity between
 these two graphs.
   
 We define the Fiedler pseudo-set~(FPS) of a graph $G$, with respect
 to a tolerance threshold $\varepsilon >0$, as follows,
\begin{equation}
    \label{eq:psuedo-Fiedler-set}
    \Lambda_2(A, \varepsilon)
    = \big\{\, \lambda_2 \!\mid  L(A+E) u = \lambda_2 u,
    \, \| E\| \leq \varepsilon \|A\| \, \big\},   
  \end{equation}
where $L$ is the normalized graph Laplacian, $E$ represents a
  perturbation or change in $A$, and $\|\cdot\|$ is a matrix norm.
  One may impose certain constraints on the perturbation matrix $E$,
  such as by a probabilistic model, or by domain-specific knowledge.
The FPS measures the variation in the Fiedler value, i.e., the network
  connectivity, in response to small changes in a graph due to random
  perturbation, editing, or rewiring.

  Graph $G$ is said $\varepsilon$-stable in connectivity if the
  dispersion of $\lambda_2$ in $\Lambda_2$ is relatively small.  When
  $\lambda_2$ is small, an $\varepsilon$-stable graph is resistant
  to division.  Among many possible ways to measure the sensitivity of
  a network connectivity, we may simply use the pair of ratios
  $t_{\min}=\min(\Lambda_2)/\lambda_2$ and
  $t_{\max} = \max(\Lambda_2)/\lambda_2$.  If $\lambda_2$ is small
  and $\min(\Lambda_2)$ is much smaller, then the graph tends to break
  apart.

  By the measure of connection sensitivity, we are able to
  definitively and quantitatively differentiate the RoK graph
  $G_{\textrm{RoK}(6,17)}$ and the cylindrical mesh
  $G_{P\boxprod C(33,3)}$ mentioned above.  We show in \cref{tab:fps}
  that for the RoK graph, $\min(\Lambda_2)$ becomes $4$ times smaller
  when one edge is removed or rewired, and it becomes zero when two
  edges are removed or rewired.  This reveals the tendency of the RoK
  graph to break apart.
In contrast, the cylindrical mesh is stable while undergoing similar
  changes, the ratio $t_{\min}$ is close to $1$, signifying a strong
  bond within the graph to stay as a whole.
Based on this analysis, we reach conclusions that are in line with
  our heuristic understanding of these graphs. We may say that
  the sensitivity analysis gives foundational support to our
  intuitive understanding of the community structure on a RoK graph.
  The FPS provides reliable information, which cannot be contained in
  a single Fiedler value, for making robust decisions in a community
  detection process.

  The FPS has broader theoretical and practical merits.

  \begin{inparaenum}[(i)]
\item The FPS is the connectivity characteristic for complex network.
  It acknowledges and regards the typical co-existence of structures and
  randomness in a complex network and measures their interplay.  The
  randomness in a complex network defies connectivity description by
  the Fiedler value at a single random instantiation.
  
\item In the context of the $\gamma$-setting with the modularity
  model, the FPS is the {\em critical region}, in place of the single
  value $\lambda_2$ in \cref{thm:Fiedler-critical-transition}. The larger the
  dispersion in $\Lambda_2$ from $\lambda_2$ is, the more uncertainty
  is in the transition between division and no-division configuration
  states.
This critical region transcends the critical transition analysis in
  \cref{sec:fiedler-transition}.

\item The FPS can be used to assess the robustness of an inferred
  community configuration, in combination with other existing
  measures~\cite{silva2022}.

\item When a complex network is stable in connectivity, by the FPS,
  its community structure can be inferred reliably from a particular
  instantiation. This includes the analysis of subnetwork structures
  in \cref{sec:subgraph-division}.
  
\item The structure of the subspace associated with the FPS can be
  used to explain or help resolve some other ambiguity or uncertainty
  issues related to graph clustering/community detection. We describe
  such a scenario in \cref{sec:additional-implications}.

\end{inparaenum}

We show in \cref{tab:fps} the Fiedler values and the Fiedler pseudo-sets for
several graphs, including a Barab\'{a}si-Albert (BA) graph~\cite{barabasi1999}
and an Erd\"{o}s-R\'{e}nyi (ER) graph~\cite{erdos1959,bollobas2001}, each with
$2000$ nodes and degree $16$ on average. The Fiedler values for the BA graph and
the ER graph are far from small, and they do not disperse much in their Fiedler
pseudo-sets. These findings align well with the existing understanding that BA
and ER graphs typically lack clear and distinct community structures.  Except
for the RoK graph, the other graphs on \cref{tab:fps} are also stable in
connectivity when subject to similar changes.

For the experimental results in \cref{tab:fps}, the Fiedler pseudo-set
$\Lambda_2(\varepsilon)$ for each graph is estimated by
$\min\{m,100\}$ random trials of edge removal or rewiring, $m$ is the
total number of edges on the graph.

\section{The sub-modularity matrix}
\label{sec:subgraph-division}

Many search algorithms for modularity maximization leverage a
highly attractive property, the separability, of the function
$Q(\Omega \mid \gamma)$ at any fixed value of the resolution parameter
$\gamma$~\cite{liu2021,floros2022}.
Let $V_1$ be a subset of vertices.  Let $H$ be the subgraph induced by $V_1$.
Denote by ${V}_2$ the complement of $V_1$ in $V$.  Let $H^{c}$ be the subgraph
induced by $V_2$.
At some algorithmic steps, the search operations can be restricted to
$H$ while the rest of the graph is fixed or undergoes changes
independently:
\begin{equation}
  \label{eq:Q-separability}
\begin{aligned} 
    Q( \Omega \mid \gamma )
    & = 
  \sum_{C_i\subset \Omega(V_1) }\alpha(C_i,C_i)
   - \gamma \, \alpha^{2}(C_i,V)
   \\
   & \phantom{xx}  + 
  \sum_{C_j \subset \Omega(V_2)  }
  \alpha(C_j,C_j) - \gamma \, \alpha^{2}(C_j,V).
\end{aligned}
\end{equation}

We introduce the augmented graph $\bar{H}$ obtained from $G$ by
joining $H$ and a graph minor of $H^{c}$. Suppose $H^{c}$ has a
cluster structure $\{C_j\}$, each cluster $C_j$ is connected. By
contracting each $C_j$ into a single node, we get a graph minor of
$H^{c}$. We detail the basic case that the entire $H^{c}$ is
contracted to a single node.  For clarity, we describe the augmented
graph $\bar{H}$ by its adjacency matrix $\bar{A}_1$,
\begin{equation}
 \label{eq:A-augmented}
  \bar{A}_1 = \left(
    \begin{array}{cc}
      A(V_1,V_1) & A(V_1,V_2)\, e   
      \\
      e^{\rm T}A( V_2, V_1) & 0 \end{array} \right)
\end{equation}
The leading matrix $A(V_1,V_1)$ is the adjacency matrix of the subgraph $H$.
The last column (row) of $\bar{A}_1$ is associated with the augmented node
$\bar{v}$, which represents the existence of all other nodes external to $H$.
The number of the nonzero elements in the column is the number of edges between
$V_1$ and $\bar{v}$.  For any $v_i\in V_1$, $A(v_i,V_2)e $ is the weight on the
edge between $v_1$ and $\bar{v}$, equal to the sum of the connection weights
between node $v_i$ and all nodes in $V_{2}$. For the extreme case that $H$ has a
single vertex $v_i$, the extended graph $\bar{H}$ is the edge graph with $v_i$
and $\bar{v}$, the edge weight is equal to $d(v_i)$. For the other extreme case
that $H$ has $(n-1)$ nodes, $\bar{A}_1=A$, i.e., the augmented graph is $G$.

When graph $H^{c}$ has a configuration of $k$ clusters $\{C_j, j=1,\dots,k\}$,
$k>1$, it can be represented by a $k$-nodes graph minor obtained by
cluster contraction. Then, graph $H$ is augmented to $\bar{H}$ by $k$
external nodes/graph minor nodes.  The external nodes are intra-linked
and weighted by the edges in the graph minor, and they are linked to
the nodes in $V_1$ by the interconnections $\alpha(v_i, C_j)$,
$v_i \in V_1$ and $C_j$ is a cluster of $H^{c}$.

The augmented graph $\bar{H}$ has the following properties.
\begin{inparaenum}[(i)]
\item It preserves the interconnection between $H$ and $H^{c}$.  As long as $G$
  is connected, the augmented graph $\bar{H}$ is connected, regardless of the
  granularity of a graph minor for representing $H^{c}$.
\item It also preserves the probabilistic transition among the nodes
  in $V_1$.  For any vertex $v_1$ in $V_1$, its volume/degree $d(v_i)$
  on $G$ is preserved on $\bar{H}$. Denote by $\bar{d}_1$ the degree
  vector on $\bar{H}$, $ \bar{d}_1 = \bar{A}_1e$.  Then,
  $\bar{d}_1(v_i) = d(v_i)$, $v_i \in V_1$. Let
  $\bar{D}_1 = \mbox{\rm diag}(\bar{d}_1 )$.  By the comparison
  between the two stochastic matrices $\bar{A}_1\bar{D}_{1}^{-1}$ and
  $AD^{-1}$, the probabilistic transition from node $v_i$ in $V_1$ to
  any other node $v_j$ in $V_1$ on $\bar{H}$ remains the same as that
  on $G$.
\item We can apply the $\gamma$-parametrized modularity model to the
  augmented graph $\bar{H}$. By the separability of the modularity
  function (\ref{eq:Q-separability}) and the properties of $\bar{H}$,
  at any value of $\gamma$ ($\bar{H}$ is independent of $\gamma$), an
  ascending/descending of $Q$ on $\bar{H}$ implies directly the same
  amount of ascending/descending of $Q$ on $G$.
If the maximal $Q$-value is at the no-partition of $\bar{H}$ or at
  the partition between $V_{1}$ and the augmented nodes, then no
  independent division of $H$ improves the modularity value, which
  implies a merge decision if $H$ results from a union of more than
  one inter-connected subgraphs.
\end{inparaenum}

We are in a position to describe the $\gamma$-parametrized modularity
matrix $\hat{B}$ for $\bar{H}$,
\begin{equation}
  \label{eq:Bhat-on-Abar}
  \hat{B}_{1}\left(\gamma \right)
  = \left(I - L(\bar{H}) \right)
  - \frac{\gamma}{e^{\rm T}\bar{d}_{1}}\, 
  \bar{d}_{1}^{1/2} \left( \bar{d}_{1}^{1/2} \right)^{\rm T},
\end{equation}
where $\bar{d}_{1} = \bar{A}_{1}e$ is the degree vector on $\bar{H}$,
$L(\bar{H})$ is the normalized Laplacian of graph $\bar{H}$.  The
formulation of (\ref{eq:Bhat-on-Abar}) is computation-friendly,
scaling the degree vector only without scaling the entire matrix
$\bar{A}_{1}$, because the Laplacian is invariant to scaling in
$\bar{A}_1$.

The theory on the spectral structure of the parametrized modularity
matrix $\hat{B}(\gamma)$ holds true for the augmented graph $\bar{H}$.
Furthermore, all the above applies to any subgraph of $\bar{H}$,
recursively, in adaption to the graph structure.  This effectively
extends the modularity spectral analysis to multi-community detection.

Additional comments.
\begin{inparaenum}[(1)]
\item The augmented nodes for external connections are necessary.
  Consider, for example, graph $G_{\rm 2stars}$, which has a single
  edge connecting two copies of a star graph with more than $3$ leaf
  nodes. Let $V_1$ be a subset of nodes with degree $1$.  Then, the
  graph $H$ induced by $V_1$ is composed of isolated nodes.  In
  contrast, with the augmented node $\bar{v}$ contracted from $H^{c}$, $\bar{H}$
  is a star graph, all nodes in $V_1$ are attached to $\bar{v}$.
\item The granularity of a graph minor for representing $H^{c}$
  affects the connectivity and cluster granularity on
  $\bar{H}$. Revisit the graph $G_{\rm 2star}$.  We can contract
  $H^{c}$ to a graph minor with two nodes, each node is associated
  with the center of a star subgraph. If $V_1$ contains leaf nodes
  from both stars, then $\bar{H}$ is also a two-star graph.
\end{inparaenum}
As a matter of fact, the augmented subgraphs are used in both
modularity model analysis and practical search algorithms for
modularity maximization~\cite{liu2021,floros2022}.

\begin{figure}[bhtp]
  \hspace*{\fill}
\includegraphics[height=0.13\textheight]{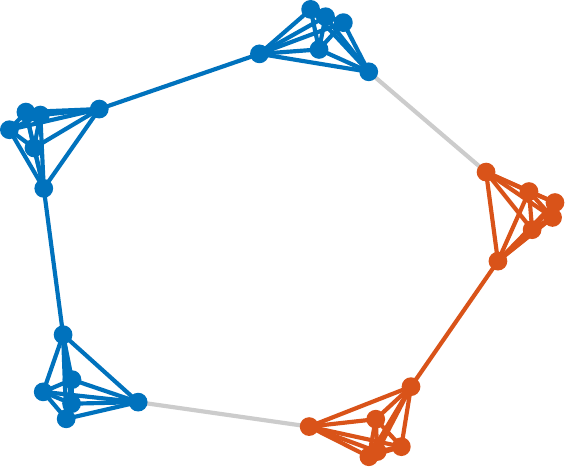}
\hspace*{\fill}
\includegraphics[height=0.10\textheight]{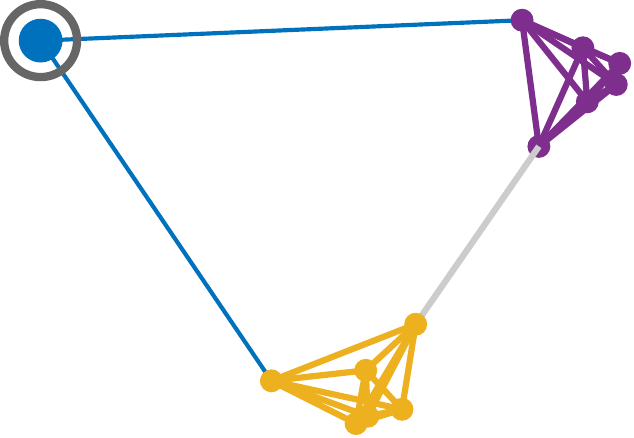}
\hspace*{\fill}
  \\[0.5em]
  \resizebox{\linewidth}{!}{\begin{tabular}{
      l
      S[table-format=1.3, table-number-alignment = right]
      S[table-format=1.3, table-number-alignment = right]
cc
cc
cc
cc
      }
      \toprule
&  &  
        & \multicolumn{2}{c}{\footnotesize $r_{w}=1$ } 
        & \multicolumn{2}{c}{\footnotesize $r_{e}=1$ } 
        & \multicolumn{2}{c}{\footnotesize $r_{w}=2$ } 
        & \multicolumn{2}{c}{\footnotesize $r_{e}=2$ } 
      \\
      \cmidrule(lr){4-5} \cmidrule(lr){6-7} \cmidrule(lr){8-9} \cmidrule(lr){10-11}
      \multicolumn{1}{l}{Graph} 
      & {$\lambda_{2}$} & {$\lambda_{3}$} 
      & {$t_{\min}$} & {$t_{\max}$}
      & {$t_{\min}$} & {$t_{\max}$}
      & {$t_{\min}$} & {$t_{\max}$}
      & {$t_{\min}$} & {$t_{\max}$}
      \\
      \midrule
      $G_{\textrm{RoK}}(6,5)$ & 0.033 & 0.033 & $\nicefrac{1}{3.47}$ & $\nicefrac{1.00}{1}$ & $\nicefrac{1}{3.61}$ & $\nicefrac{1.00}{1}$ & $\nicefrac{1}{3.37}$ & $\nicefrac{1.54}{1}$ & $\nicefrac{1}{\infty}$ & $\nicefrac{1.02}{1}$  
      \\
      $\bar{H}_{\textrm{red}}$ & 0.079 & 0.716 & $\nicefrac{1}{2.12}$ & $\nicefrac{1.70}{1}$ & $\nicefrac{1}{2.79}$ & $\nicefrac{1.04}{1}$ & $\nicefrac{1}{1.67}$ & $\nicefrac{2.46}{1}$ & $\nicefrac{1}{\infty}$ & $\nicefrac{1.08}{1}$  
      \\
      \bottomrule
    \end{tabular}
  }
  \caption{ {\tt Left}: The ring-of-clique (RoK) graph
    $G_{\textrm{RoK}(6,5)}$ and its division into two subgraphs
    $H_{\rm blue}$ and $H_{\rm red}$ in red and blue, respectively.
    The Fiedler value of $G_{\rm RoK(6,5)}$ is of multiplicity $2$ and
    relatively small to the average of the Laplacian values,
    It decreases with one edge removal/rewiring and becomes zero with two-edges
    removal/rewiring.
{\tt Right.} The augmented graph $\bar{H}_{\rm red}$ and the
    further split of $H_{\rm red}$ into two subgraphs in purple and
    yellow, respectively.  The circled blue node denotes the contracted graph $H_{\rm blue}$.  The two edges between
    $H_{\rm blue}$ and $H_{\rm red}$ are maintained in
    $\bar{H}_{\rm red}$.  The Fiedler value of $\bar{H}_{\rm red}$ is
    simple and small. It decreases with one-edge removal/rewiring and
    becomes zero with two-edges removal/rewiring.}
\label{fig:contracted-subgraph-of-RoK}
\end{figure}

We show in \cref{fig:contracted-subgraph-of-RoK} the RoK graph
$G_{\textrm{RoK}(6,5)}$, its division into two subgraphs $H_{\rm red}$
and $H_{\rm blue}$, and the augmented graph $\bar{H}_{\rm red}$.  The
augmented graph $\bar{H}_{\rm red}$ has the external node $\bar{v}$
(annotated by the circled blue marker) and maintains the two edges
between $H_{\rm blue}$ and $H_{\rm red}$.
The Fiedler value of $G_{\rm RoK(6,5)}$ is relatively small compared
to the average $n/(n-1)$ ($n=30$) of the Laplacian values for the
graph.  Furthermore, by the FPS descriptor
$t_{\min} = \min(\Lambda_2)/\lambda_2$, the Fiedler value becomes $3$
times smaller over one-edge removal/rewiring, and becomes $0$ over
two-edges removal.  The division of $G_{\textrm{RoK}(6,5)}$ into
$H_{\rm red}$ and $H_{\rm blue}$ is justified. By similar arguments,
the split of $\bar{H}_{\rm red}$ is certified.

\begin{figure}[htbp]
   \centering
    \hspace*{\fill}
    \includegraphics[height=0.14\textheight]{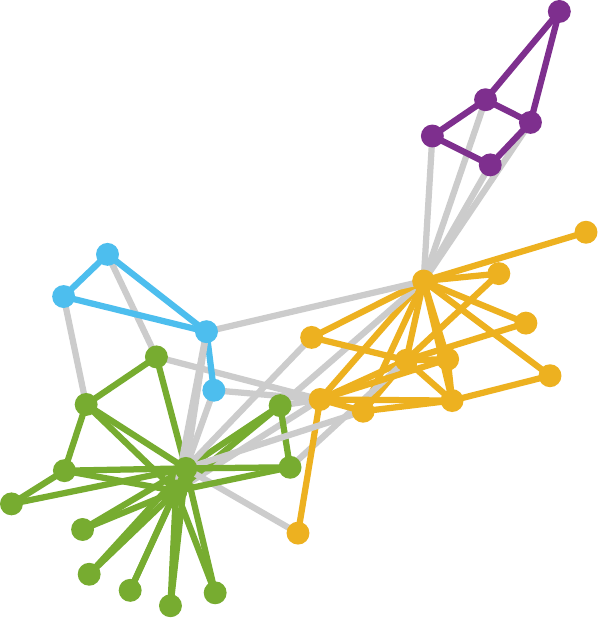}
\hspace*{\fill}
\includegraphics[height=0.14\textheight]{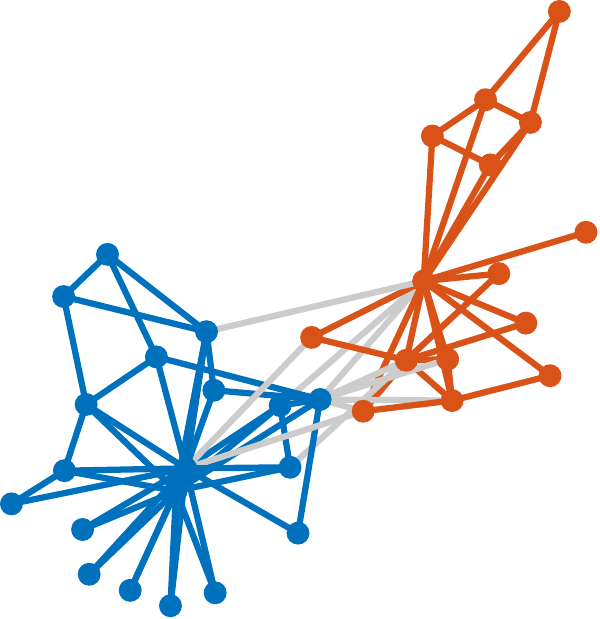}
\hspace*{\fill}
\\[0.5em]
  \resizebox{\linewidth}{!}{\begin{tabular}{
      l
      S[table-format=1.3, table-number-alignment = right]
      S[table-format=1.3, table-number-alignment = right]
cc
cc
cc
cc
      }
      \toprule
      &  &  
        & \multicolumn{2}{c}{\footnotesize $r_{w}=1$ } 
        & \multicolumn{2}{c}{\footnotesize $r_{e}=1$ } 
        & \multicolumn{2}{c}{\footnotesize $r_{w}=2$ } 
        & \multicolumn{2}{c}{\footnotesize $r_{e}=2$ } 
      \\
      \cmidrule(lr){4-5} \cmidrule(lr){6-7} \cmidrule(lr){8-9} \cmidrule(lr){10-11}
      \multicolumn{1}{l}{Graph} 
      & {$\lambda_{2}$} & {$\lambda_{3}$} 
      & {$t_{\min}$} & {$t_{\max}$}
      & {$t_{\min}$} & {$t_{\max}$}
      & {$t_{\min}$} & {$t_{\max}$}
      & {$t_{\min}$} & {$t_{\max}$}
      \\
      \midrule
      $G_{\textrm{karate}}$ & 0.132 & 0.287 & $\nicefrac{1}{1.11}$ & $\nicefrac{1.30}{1}$ & $\nicefrac{1}{1.18}$ & $\nicefrac{1.06}{1}$ & $\nicefrac{1}{1.20}$ & $\nicefrac{1.46}{1}$ & $\nicefrac{1}{1.18}$ & $\nicefrac{1.07}{1}$  
      \\
      $\bar{H}_{\textrm{blue}}$ & 0.379 & 0.419 & $\nicefrac{1}{1.20}$ & $\nicefrac{1.10}{1}$ & $\nicefrac{1}{1.17}$ & $\nicefrac{1.08}{1}$ & $\nicefrac{1}{1.23}$ & $\nicefrac{1.15}{1}$ & $\nicefrac{1}{1.27}$ & $\nicefrac{1.12}{1}$  
      \\
      $\bar{H}_{\textrm{red}}$ & 0.235 & 0.705 & $\nicefrac{1}{1.31}$ & $\nicefrac{1.40}{1}$ & $\nicefrac{1}{1.30}$ & $\nicefrac{1.11}{1}$ & $\nicefrac{1}{1.36}$ & $\nicefrac{1.68}{1}$ & $\nicefrac{1}{1.78}$ & $\nicefrac{1.19}{1}$  
      \\
      \bottomrule
    \end{tabular}
  }
  \caption{{\tt Left:} Zahary's karate club graph~\cite{zachary1977}
    is divided into four communities by the modularity model at
    $\gamma=1.0$.
{\tt Right:} The same graph is split into two subgraphs,
    $H_{\rm blue}$ in blue and $H_{\rm red} $ in red, by the Fiedler
    vector of the club graph. 
The table provides the Fiedler values and the FPS descriptors
    $(t_{\min}, t_{\max})$ for the club graph, the augmented subgraph
    $\bar{H}_{\rm red}$ and the augmented subgraph
    $\bar{H}_{\rm blue}$. }
  \label{fig:karate-cut}
\end{figure}

In \cref{fig:contracted-subgraph-of-RoK}, we show Zahary's karate club
graph~\cite{zachary1977} and two different community configurations.  The
configuration with four communities is by the modularity maximization at
$\gamma=1$.  The configuration with two communities is by the Fiedler spectral
information of the graph and the two (augmented) subgraphs by the Fiedler cut.
The Fiedler values of the subgraphs become much larger, they do not decrease
much by random removal/rewiring of one or two edges.

Our analysis and experimental results have shown that the Fiedler
values of (augmented) subgraphs vary from each other and from the
Fiedler value of graph $G$.  This brings to light another issue that
the $\gamma$-setting also affects subgraph split or merge decisions,
consequently, the recognition or misidentification of sub-community
structures.

\section{Discussion}
\label{sec:additional-implications}

Recall the standard modularity matrix $\hat{B}(\gamma)$ as defined in
(\ref{eq:B-Bhat-gamma}). The modularity spectral analysis across the
variation of the resolution parameter $\gamma$ presented in the
preceding sections is grounded in two key connections that we have
established.
First, we identify the $\gamma$-parametrized modularity function
$Q(\Omega \!\mid\! \gamma)$ over all possible binary graph partitions
with the quadratic expression $y^{\rm T}\hat{B}(\gamma)y$ over the
discrete set $Y = \{y \mid y = D^{1/2}s, s = \{1,-1\}^s\}$.
Secondly, as the discrete set $Y$ is fittingly embedded within the
unit sphere $S = \{u \mid u^{\rm T}u = 1\}$, the quadratic expression
of the modularity function $Q$ can be readily extended to its
continuous counterpart, the Rayleigh quotient
$R(u \mid \hat{B}(\gamma))$ on the sphere $S$.
Understanding $R(u \!\mid\! \hat{B}(\gamma))$ gives us invaluable
insights into the modularity maximization.
The behavior and properties of $R(u \mid \hat{B}(\gamma))$ on the
sphere $S$ can be entirely determined by the spectral structure of the
kernel matrix $\hat{B}(\gamma)$.

Our modularity spectral analysis offers a unique perspective on the
relationship between the graph spectral structure and the modularity
spectral structure.  In a nutshell, the entire Laplacian spectral
structure is in the modularity matrix $\hat{B}(\gamma)$, across all
variations in $\gamma$. Roughly speaking, while every Laplacian vector
is a scalar (heat) map over the vertex set at a specific Laplacian
eigenvalue (energy level), a maximum eigenvector of $\bar{B}(\gamma)$
induces a community map of the graph/network.

We have uncovered that every maximum eigenvector of the modularity
matrix $\hat{B}(\gamma)$, $\gamma>0$, lies in the invariant subspace
of the normalized graph Laplacian $L(G)$ associated with the two
smallest Laplacian eigenvalues: $\lambda_1=0$ and $\lambda_2>0$.  When
$\gamma\neq \lambda_2$, the maximum vectors of the modularity matrix
$\hat{B}(\gamma)$ are either in the null space of $L(G)$ or in the
Fiedler space as defined in \cref{sec:fiedler-connection}.
We provide a few more comments below.

\begin{inparaenum}
\item[\emph{(A) The proxy criteria \& configuration candidates.}]  The
  Fiedler connection analysis not only provides profound insights into
  the behaviors and properties of the modularity model and the
  maximization process, it also promises a potential impact on the
  development of effective criteria for $\gamma$-setting as well as
  for efficient algorithmic search of maximal-modularity
  configurations.  Specifically, we may use the Fiedler value
  $\lambda_2$ and the Fiedler pseudo-set $\Lambda_2(\varepsilon)$ of a
  graph, or a subgraph, to assess the merge or split results at a
  specified $\gamma$ value and provide valuable feedback for adaptive
  $\gamma$-setting.  Subsequently, the maximum vectors of matrix
  $\hat{B}(\gamma)$ propose configuration candidates.  All maximum
  vectors reside in the null space or/and the Fiedler space.
  Computationally, the eigenvector associated with the zero eigenvalue
  of the Laplacian is unique and it is $d^{1/2}$, readily available.
  The Fiedler value $\lambda_2$ can be easily obtained numerically. If
  $\lambda_2$ is simple, the Fiedler vector is unique, and it is also
  easily obtainable.

  In practice, the Fiedler vectors have long been used in spectral
  clustering algorithms for image segmentation and some other
  domain-specific graph clustering applications, and the search over
  a continuous domain is combined with the search over a discrete
  domain~\cite{shi2000,yu2003}. 

\item[\emph{(B) The Fiedler space \& graph symmetry.}]
When the Fiedler space is dominant in $\hat{B}(\gamma)$ and multi-dimensional,
  the variety in graph division is great and not to be overlooked.  The
  existence of many different Fiedler cuts on the same graph indicates the
  presence of homogeneous or symmetric substructures in the graph. Consider a
  few particular graphs with symmetric structures -- $K_n$, $G_{\rm cube(k)}$,
  and $G_{\rm buckyball}$.  The multiplicity of $\lambda_2$ reaches the highest
  at $(n\!-\!1)$ for the clique $K_n$, it is equal to $k$ for the
  $k$-dimensional hypercube, and equal to $3$ for the buckyball $G_{\rm
  buckyball}$.

  Although none of these specific graphs with symmetric structures,
  $K_n$, $G_{\rm cube(k)}$, $G_{\rm buckyball}$ and
  $G_{P\boxprod C(p,q)}$, should be divided, the presence of symmetric
  structures does not preclude the possibility of identifying
  meaningful graph splits. Consider, for instance, the RoK graph,
  $G_{\textrm{RoK}(p,q)}$, $p,q >4$. The multiplicity of its Fiedler
  value is greater than one, reflecting, indeed, the symmetric
  substructures. However, each clique $K_p$ in $G_{\textrm{RoK}(p,q)}$
  is certified as a closely tied community/cluster, as we have discussed in
  \cref{sec:FPS-transition,sec:subgraph-division}.

  The multiplicity of the Fiedler value highlights the broader
  phenomena where $\lambda_2$, which may be simple, is followed
  closely by $\lambda_3$ or more than one other eigenvalues.  For
  instance, the Fiedler value of the cylindrical mesh
  $G_{P\boxprod C(p,q)} $ is simple, and it is tightly followed by the
  other eigenvalues.  Such spectral cluster phenomena indicate
  (nearly) homogeneous or symmetric substructures within a graph.

  The important message is this. The spectral cluster at the Fiedler
  value, if not sensitive to pertrubation, signifies the strength in
  bounding. When the Fiedler value is small, the graph is thin and
  strong like a graphene mesh. This finding raises the new issue on
  the community bounding strength, in addition to the issue with
  community size/resolution. The answer is again in the Fiedler
  pseudo-set $\Lambda_2$.

  There are other complicated issues when the Fiedler space is
  dominant and multi-dimensional. All maximum vectors of $\hat{B}$
  arrive at the same maximal value $\beta_{\max}$ of the continuous
  quadratic function $u^{\rm T}\hat{B}(\gamma)\,u$. However, their
  sign indicators are not necessarily equal in the modularity value,
  $y(s)^{\rm T}\hat{B}(\gamma)\, y(s) = s^{\rm T} B(\gamma)s$, with
  $s=\mbox{\rm sign}(u)$. Ideally, we would like to locate a
  unit-length maximum vector with the maximal value of
  $ u^{\rm T}y({\rm sign}(u))$, the similarity between $u$ and its
  counterpart $y(\mbox{\rm sign}(u))$.
Computationally, an eigenvector(s) associated with the
  $\min( \Lambda_2(\varepsilon) ) $, with an $\varepsilon$-departure
  from the symmetry would help single out a good candidate
  configuration among innumerably many incidental combinations. This
  configuration would significantly reduce the combinatorial search
  cost.

\end{inparaenum}

\begin{inparaenum}
\item[] \emph{(C) Resolution parameter variation.}   
By our analysis, it is necessary to consider the consistency and the variation
  of the maximum cluster configurations in response to changes in  $\gamma$.
  Using the modularity model with a single value of $\gamma$, or even a narrow
  range of $\gamma$-values, is limited in detecting community structures with
  more than one heterogeneous communities.
It is recently reported in~\cite{floros2018, floros2022,
    floros2022a, liu2021,liu2021c} that a network with stable
  distinctive community structures is affiliated with more than one
  $\gamma$ regions, each region signifying the dominance of a cluster
  configuration at an intrinsic level of cluster granularity.

\item[] \emph{(D) Model transformation.}  
With its beautiful simplicity and other attractive properties, the
  $\gamma$-parametrized modularity model remains broadly deployed, among many
  other existing models, for community detection.
Our analysis exposes notwithstanding an inherent restriction within
  the framework across the entire $\gamma$-variation. Specifically,
  the spectral modulation mechanism affected by the modularity matrix
  $\hat{B}(\gamma)$ is constrained to the choice between the Fiedler
  space and the null space of the graph Laplacian.  The finer
  structures are captured by other eigenvectors at the lower end of
  the Laplacian spectrum.  It is not inconceivable to relax or remove
  the constraint without confining ourselves to the use of
  divide-and-conquer hierarchies.  It is possible to modulate or
  reshape the model structure over a range of the graph spectrum
  beyond the two smallest eigenvalues, $ \{\lambda_1,\lambda_2\} $,
  with respect to a particular class of graphs of interest. The
  normalized Laplacian spectral structure decomposition is identified
  with the generalized eigenvalue problem $Ax=\lambda Dx$.  The graph
  structures may also be analyzed or characterized by the generalized
  eigenvalue problem with a different symmetric, positive definite
  mass matrix $M$, $Ax = \lambda Mx$, or even a non-linear generalized
  eigenvalue problem~\cite{pasadakis2022,kawamoto2023}. It is likely that existing,
  emerging, and competing models for community detection could be
  interpreted as different spectral modulation mechanisms.
  
\end{inparaenum}

\appendix
\section{Quadratic form of modularity cuts}
\label{sec:appendix-proofs}

\begin{proof}[Proof of (\ref{eq:Q-B-gamma})]
Denote by $e(C)$ the vector that is constant $1$ with length $|C|$,
  $C\subset V$. For any configuration $\Omega$ with two
  clusters, $\Omega = \{C_1,C_2\}$, the adjacency matrix
  can be reordered so that
\vspace*{1em}
\[
    A = \left( \begin{array}{cc} A_{11} & A_{12}
        \\ A_{21} & A_{22} \end{array} \right),
    \quad
    A_{ij} = A( C_i,C_j),
    \quad i, j = 1,2. 
  \]
\vspace*{2em}
The indicator of $\Omega$ can then be described as follows, 
  \[
  s^{\rm T}(\Omega) = [ e^\mathrm{T}(C_1), \, -e^\mathrm{T}(C_2) ]. 
  \]

  \vspace*{1em}

  Assume that $e^{\rm T}Ae=1$. Then, the degree vector $d=Ae$ is
  stochastic. Let $d_{i} = d(C_i)$, $i=1,2$.
  We have
  \\
{ \footnotesize 
    \begin{align*} 
      & s^\mathrm{T}(\Omega) B(\gamma) s(\Omega) = 
\\
& \begin{bmatrix}
        e^\mathrm{T}(C_1) & -e^\mathrm{T}(C_2)
      \end{bmatrix}
\begin{bmatrix}
        A_{11} - \gamma d_ {1} d_{1}^\mathrm{T} 
          & A_{12} - \gamma d_ {1} d_{2}^\mathrm{T}
        \\[0.5em]
        A_{21} - \gamma d_ {2} d_{1}^\mathrm{T} 
          & A_{22} - \gamma d_ {2} d_{2}^\mathrm{T}
      \end{bmatrix}
\begin{bmatrix}
        e(C_1)
        \\[0.5em]
        -e(C_2)
      \end{bmatrix}.
     \end{align*}
  }
Using the $\alpha$-notation of (\ref{eq:Q-alpha}), we have
  $e^{\rm T}(C_i)d(C_i) = \alpha(C_i,V)$, $i=1,2$.  The two terms in
  $s^{\rm T}B(\gamma)s$ related to the diagonal blocks of $B(\gamma)$
  sum to
  \[ \footnotesize \sum_{i=1,2} \alpha(C_i,C_i) - \gamma
    \alpha^{2}(C_i,V) = Q(\Omega \mid \gamma).
  \]
  The two terms related to the off-diagonal blocks of $B(\gamma)$ are
  equal to each other by the symmetry and sum to 
  \begin{align*}
    2&\big( \gamma\, \alpha(C_1,V) \alpha(C_2,V) - \alpha(C_1,C_2) \big)
    \\
    &=
    \alpha(C_1,C_1) + \alpha(C_2,C_2) - \big( \underbrace{ \alpha(C_1, V) + \alpha(C_2, V) }_{=1} \big) + 
    \\ &\phantom{=} + \gamma \big( \, \underbrace{ \alpha(C_1, V) + \alpha(C_2, V) }_{=1} - \alpha^2(C_1,V) - \alpha^2(C_2,V) \, \big)
    \\
    & = Q(\Omega \mid \gamma) + \gamma - 1,
  \end{align*}
where the following equalities are used
\begin{align*} 
     \alpha(C_1,C_2) &= \alpha(C_2, C_1),
    \\
    \alpha(C_1,V) &+ \alpha(C_2,V) = 1,
    \\
    \alpha(C_1, C_2) &= \alpha(C_1,V) - \alpha(C_1,C_1) 
    \\
    &= \alpha(C_2,V) - \alpha(C_2,C_2).
  \end{align*} 
Thus, we have 
  \[
    Q(\Omega \mid \gamma) = \dfrac{1}{2} s^\mathrm{T}(\Omega)
    B(\gamma) s(\Omega) + \left( \dfrac{1 - \gamma}{2} \right).
  \]
  This equality extends to the all-in-one configuration $\Omega_{\vee}$
  because $  Q(\Omega_{\vee} \mid \gamma) = 1-\gamma$.
\end{proof}

\end{document}